\pdfoutput=1

\documentclass{article}

\usepackage{arxiv}

\usepackage{listings}
\usepackage{amssymb,amsfonts, amsmath}
\usepackage{clrscode3e}
\usepackage{amsthm} 
\usepackage[linesnumbered,ruled]{algorithm2e}
\usepackage[utf8]{inputenc}
\usepackage{subcaption}
\usepackage{xspace}
\usepackage{graphicx}
\usepackage{textcomp}
\usepackage{xcolor}
\usepackage{multirow}

\newtheorem{definition}{Definition}
\newtheorem{theorem}{Theorem}
\newtheorem{lemma}{Lemma}
\newtheorem{fact}{Fact}

\newcommand{\BO}[1]{O\left(#1\right)}

\newcommand{\BT}[1]{\Theta\left(#1\right)}
\newcommand{\alg}[1]{\texttt{#1}}

\def\cost{\mbox{\rm cost}}
\def\OPT{\mbox{\rm OPT}}
\def\out{\mbox{\rm out}}

\title{Distributed k-Means with Outliers in General Metrics}

\author{Enrico Dandolo\\
University of Padova\\
 Padova,  Italy\\
\texttt{enrico.dandolo.1@studenti.unipd.it}\\
\And
Andrea Pietracaprina\\
University of Padova\\
 Padova,  Italy\\
\texttt{andrea.pietracaprina@unipd.it}\\
\And
Geppino Pucci\\
University of Padova\\
 Padova,  Italy\\
\texttt{geppino.pucci@unipd.it}\\
}

\begin{document}
\maketitle

\begin{abstract}
Center-based clustering is a pivotal primitive for unsupervised learning and data analysis. A popular variant is undoubtedly the k-means problem, which, given a set $P$ of points from a metric space and a parameter $k<|P|$, requires to determine a subset $S$ of $k$ centers minimizing the sum of all squared distances of points in $P$ from their closest center.  A more general formulation, introduced to deal with noisy datasets, features a further parameter $z$ and allows up to $z$ points of $P$ (outliers) to be disregarded when computing the aforementioned sum. We present a distributed coreset-based 3-round approximation algorithm for k-means with $z$ outliers for general metric spaces, using MapReduce as a computational model. Our distributed algorithm requires sublinear local memory per reducer, and yields a solution whose approximation ratio is an additive term $O(\gamma)$ away from the one achievable by the best known sequential (possibly bicriteria) algorithm, where $\gamma$ can be made arbitrarily small.  An important feature of our algorithm is that it obliviously adapts to the intrinsic complexity of the dataset, captured by its doubling dimension $D$.  To the best of our knowledge, no previous distributed approaches were able to attain similar quality-performance tradeoffs for general metrics.

\end{abstract}
\section{Introduction} \label{sec-intro}
Clustering is a fundamental primitive for data analysis and
unsupervised learning, with applications to such diverse domains as
pattern recognition, information retrieval, bioinformatics, social
networks, and many more \cite{HennigMMR15}. Among the many approaches
to clustering, a prominent role is played by \emph{center-based
  clustering}, which aims at partitioning a set of data items into $k$
groups, where $k$ is an input parameter, according to a notion of
similarity modeled through a metric distance over the data. Different
variants of center-based clustering aim at minimizing different
objective functions. The  \emph{k-means}
problem is possibly the most popular variant of center-based
clustering. Given a set $P$ of points in a general metric space and a
positive integer $k <  |P|$, the discrete version of the problem requires to determine a 
subset $S\subset P$ of $k$ points, called \emph{centers}, so that
the sum of all squared distances of the points of $P$ from their
closest center is minimized. (In Euclidean spaces, centers may be
chosen also outside the set $P$, giving rise to the continuous version
of k-means, admitting a wider spectrum of feasible solutions.)

Since the objective function of k-means involves squares of distances, the optimal solution is at
risk of being  impacted by few ``distant'' points,
called \emph{outliers}, which may severely bias the optimal center selection towards reducing such distances.  In fact, the presence of outliers is inevitable
in large datasets, due to the present of  points which are artifacts of data
collection, either representing noisy measurements or simply erroneous
information.  To cope with this limitation, k-means admits a
heavily studied robust formulation that takes into account
outliers \cite{CharikarKMN01}: when
computing the objective function for a set of $k$ centers, the $z$
largest squared distances from the centers are not included in the sum, where 
$z<|P|$
is an additional input parameter representing a tolerable level of
noise. This formulation of the problem is known as  \emph{k-means  with
$z$ outliers}.

There is an ample and well-established literature on sequential strategies for different instantiations of center-based clustering, with and without outliers.
However, with the advent of big data, the high volumes that need to be processed often rule out the use of unscalable, sequential strategies. 
Therefore, it is of paramount importance to devise efficient clustering strategies tailored to  typical distributed computational frameworks for big data processing (e.g., MapReduce \cite{DeanG08}). 
The primary objective of this paper is to devise scalable, distributed  strategies for  discrete k-means  with $z$ outliers for general metric spaces.

\subsection{Related Work}\label{sec:previous-work}
The body of literature on solving k-means without outliers
sequentially is huge \cite{AwasthiB15}. The best sequential algorithms
to date for the discrete case on general metrics are the deterministic
$(9+\epsilon)$-approximation algorithm of \cite{AhmadianNSW20}, or the
randomized PTAS of \cite{CohenFS21} for spaces of constant doubling
dimension. A simpler and faster randomized option is the
\alg{k-means++} algorithm of \cite{ArthurV07}, whose approximation
ratio, which is $O(\log k)$ in expectation, can be lowered to a
constant by running the algorithm for $\rho k$ centers, with $\rho =
O(1)$ \cite{Wei16}.  For the distributed case, a 3-round
MapReduce algorithm for k-means is presented in
\cite{MazzettoPP19}. For arbitrarily small $\gamma >0$, the algorithm attains an approximation ratio
which is a mere $\BO{\gamma}$ term away from the best sequential
approximation attainable for the weighted variant of the problem.

A considerable number of sequential algorithms have also been proposed
for k-means with $z$ outliers. Here, we report only on the works most
relevant to our framework, and refer to \cite{DeshpandeKP20} for a
more detailed overview of the literature.  In \cite{GuptaKLMV17}, a
randomized local search strategy is described, which runs in time
$\BO{|P|z+(1/\epsilon)k^2(k+z)^2\log(|P|\Delta)}$, yielding a
274-approximate bicriteria solution with $k$ centers and $O((1/\epsilon)kz\log(|P| \Delta))$
outliers, where $\Delta$ is the ratio between the maximum and minimum
pairwise distances.  For spaces of doubling dimension $D$,
\cite{FriggstadKRS19} devises a different (deterministic) local search
strategy yielding a bicriteria solution with $(1+\epsilon)k$ centers and $z$
outliers achieving approximation $1+\BO{\epsilon}$, in time 
$\BO{(k/\epsilon)|P|^{(D/\epsilon)^{\Theta(D/\epsilon)}} \log (|P|\Delta)}$.
Finally, the LP-based approach of
\cite{KrishnaswamyLS18} yields the first non-bicriteria  solution
featuring an expected $53.002\cdot (1+\epsilon)$-approximation in time
$|P|^{\BO{1/\epsilon^3}}$.

The literature on distributed approaches to k-means with outliers is more scant. The simple, sequential coreset-based strategy of \cite{StatmanRF20} can be easily made into a 2-round MapReduce algorithm yielding a solution featuring a nonconstant 
$\BO{\log(k+z)}$ approximation and  local memory $\sqrt{|P|(k+z)}$. In \cite{GuhaLZ19}, an LP-based algorithm is developed for the coordinator model, yielding a $\BO{1+1/\epsilon}$-approximate bicriteria solution, with an excess factor $(1+\epsilon)$ either in the number of outliers or in the number of centers, using $\tilde{O}(Lk+z)$ communication words, where $L$ is the number of available workers. In the coordinator model, better bounds have been obtained for the special case of Euclidean spaces in \cite{LiG18,ChenAZ18}.

\subsection{Our Contribution}
We present a scalable coreset-based distributed MapReduce algorithm for
k-means with $z$ outliers, targeting the solution of very large
instances from general metrics. The algorithm first computes,
distributedly, a coreset of suitably selected input points which act
as representatives of the whole input, where each coreset point is
weighted in accordance to the number of input points it represents.
Then, the final solution is computed by running on the coreset an
$\alpha$-approximate sequential algorithm for the weighted variant of
the problem.  Our approach is flexible, in the sense that the
final solution can also be extracted through a sequential bicriteria
algorithm returning a larger number $\rho k$ of centers and/or
excluding a larger number $\tau z$ of outliers. Our distributed
algorithm features an approximation ratio of $\alpha+\BO{\gamma}$,
where $\gamma$ is a user-provided accuracy parameter which can be made
arbitrarily small.  The algorithm requires 3 rounds and a local memory
at each worker of size $\BO{\sqrt{|P| (\rho k+\tau z)}
  (c/\gamma)^{2D}\log^2 |P|}$, where $c$ is a constant and $D$ is the
doubling dimension of the input. For reasonable configurations of the
parameters and, in particular, low doubling dimension, the local space
is substantially smaller than the input size. It is important to
remark that the algorithm is \emph{oblivious} to $D$, in the sense
that while the actual value of this parameter (which is hard to
compute) influences the analysis, it is not needed for the algorithm
to run. As a proof of concept, we describe how the sequential bicriteria algorithms by \cite{GuptaKLMV17} and \cite{FriggstadKRS19} can be extended to handle
weighted instances, so that, when used within our MapReduce algorithm
allow us to get comparable constant approximations in distributed
fashion.

We remark that the main contributions of our algorithm are: (i) its
simplicity, since our coreset construction does not require multiple
invocations of complex, time-consuming sequential algorithms for
k-means with outliers (as is the case in \cite{GuhaLZ19}); and (ii)
its versatility, since the scheme is able to exploit any sequential
algorithm for the weighted case (bicriteria or not) to be run on the
scaled-down coreset with a minimal extra loss in accuracy. In fact, to
the best of our knowledge, ours is the first distributed algorithm
that can achieve an approximation arbitrarily close to the one of the
best available sequential solution, either exact of
bicriteria. Finally, we observe that our MapReduce algorithm can solve
instances of the problem without outliers with similar approximation
guarantees, and its memory requirements imrpove substantially upon
those of \cite{MazzettoPP19}.

\paragraph{\bf Organization of the paper.}
Section~\ref{sec-prelim} contains the main definitions and some
preliminary concepts.  Section~\ref{sec:algorithm} describes a
simplified coreset construction (Subsection~\ref{subsec-coreset}),
the full algorithm (Subsection~\ref{subsec-complete}), and a sketch of
a more space-efficient coreset construction, which yields our main
result (Subsection~\ref{subsec-improved}).  Finally,
Section~\ref{sec-wkmeans} discusses the extension of the algorithms in
\cite{GuptaKLMV17} and \cite{FriggstadKRS19} to handle weighted
instances. Section~\ref{sec:conclusions} provides some final remarks.

\section{Preliminaries} \label{sec-prelim}
Let $P$ be a set of points from a metric space with 
distance function $d(\cdot,\cdot)$. For any point $p \in P$ and subset
$S \subseteq P$, define the distance between $p$ and $S$ as
$d(p,S) = \min_{q \in S} d(p,q)$. Also, we let $p^S$ denote a point of
$S$ closest to $p$, that is, a point such that $d(p,p^S)=d(p,S)$, with
ties broken arbitrarily. The discrete \emph{k-means}
problem requires that, given $P$ and an integer $k < |P|$, a set $S
\subset P$ of $k$ \emph{centers} be determined, minimizing the
cost function $\cost(P,S)=\sum_{p \in P} d(p,S)^2$.  We
focus on a robust version of discrete k-means, known in the literature
as \emph{k-means with $z$ outliers}, where, given an additional
integer parameter $z\leq |P|$ , we seek a set $S
\subset P$ of $k$ {centers} minimizing the cost function
$\cost(P \backslash \out_z(P,S),S)$, where $\out_z(P,S)$ denotes the set
of $z$ points of $P$ farthest from $S$, with ties broken
arbitrarily. We let
$\OPT_k(P)$ (resp., $\OPT_{k,z}(P)$) denote the cost of the optimal
solution of  k-means  (resp., k-means  with $z$
outliers) on  $P$.

\begin{table}[tbp]
\scriptsize
\begin{center}
\begin{tabular}{|rcl|}
\hline
$\cost(P,S)$ & = & $\sum_{p \in P} d(p,S)^2$ \\[0.1cm]
$\OPT_k(P)$ & = & $\min_{S \subset P, |S|=k} \cost(P,S)$\\[0.1cm]
$\out_z(P,S)$ & = & $z$ points of $P$ farthest from $S$ \\[0.1cm]
$\OPT_{k,z}(P)$ & = & $\min_{S \subset P, |S|=k} \cost(P\backslash\out_z(P,S),S)$\\[0.1cm] 
$\cost(P,\mathbf{w},S)$ & = & $\sum_{p \in P} w_p d(p,S)^2$ \\[0.1cm] 
$\OPT_k(P,\mathbf{w})$ & = & 
$\min_{S \subset P, |S|=k} \cost(P,\mathbf{w},S)$\\[0.1cm]  
$\OPT_{k,z}(P,\mathbf{w})$ & = & 
$\min_{S \subset P, |S|=k} \cost(P,\mathbf{\hat{w}},S)$, where 
$\mathbf{\hat{w}}$ is obtained from $\mathbf{w}$ \\[0.1cm]  
&& by scaling 
$z$ units from points of $P$ farthest from $S$\\ \hline
\end{tabular}
\end{center}
\caption{Notations used throughout the paper: $P$ is a set of $|P|$ points,
$S$ is a subset of $P$, and $0 < z < n$ is an integer parameter.} \label{tab:notation}
\end{table}

The following two facts state technical properties that will be needed
in the analysis. 
\begin{fact} \label{fact:opt}
For every $k,z >0$ we have
$\OPT_{k+z}(P) \leq \OPT_{k,z}(P)$.
\end{fact}
\begin{proof}
Let $S^*$ be the optimal solution of k-means with $z$ outliers on $P$,
that is, such that $\cost(P \backslash \out_z(P,S^*),S) = \OPT_{k,z}(P)$,
and let $\bar{S} = S^* \cup \out_z(P,S^*)$. Since $|\bar{S}| \leq k+z$,
we have that 
\[
\OPT_{k+z}(P) \leq \cost(P,\bar{S}) \leq \cost(P \backslash \out_z(P,S^*),S) = \OPT_{k,z}(P).
\]
\end{proof}
\begin{fact} \label{fact:triangle}
For any $p,q,t \in P$, $S \subseteq P$, and $c>0$, we have:
\begin{eqnarray*}
d(p,S) & \leq & d(p,q)+d(q,S) \\ 
d(p,t)^2 & \leq & (1+c)d(p,q)^2+(1+1/c)d(q,t)^2. 
\end{eqnarray*}
\end{fact}
\begin{proof}
The first inequality follows since $d(p,S) = d(p,p^S) \leq d(p,q^S) \leq d(p,q)+d(q,q^S)=d(p,q)+d(q,S)$.
The second inequality  follows since
$d(p,t)^2 \leq (d(p,q)+d(q,t))^2 = d(p,q)^2+d(q,t)^2+2d(p,q)d(q,t)$
and $2d(p,q)d(q,t) \leq (1/c)d(p,q)^2+cd(q,t)^2$,
where the latter inequality holds since $((1/\sqrt{c})d(p,q)-\sqrt{c}d(q,t))^2\geq 0$.
\end{proof}
In the \emph{weighted} variant of k-means, each point $p\in P$ carries
a positive integer weight $w_p$. Letting $\mathbf{w}:P\rightarrow
\mathbb{Z}^+$ denote the weight function, the problem requires to
determine a set $S \subset P$ of $k$ centers minimizing the cost
function $\cost(P,\mathbf{w},S)=\sum_{p \in P} w_p
d(p,S)^2$. Likewise, the weighted variant of k-means with $z$ outliers
requires to determine $S \subset P$ which minimizes the cost function
$\cost(P,\mathbf{\hat{w}},S)$, where $\mathbf{\hat{w}}$ is obtained
from $\mathbf{w}$ by scaling the weights associated with the points of
$P$ farthest from $S$, progressively until exactly $z$ units of
weights overall are subtracted (again, with ties broken arbitrarily).
We let $\OPT_k(P,\mathbf{w})$ and $\OPT_{k,z}(P,\mathbf{w})$ denote
the cost of the optimal solutions of the two weighted variants above,
respectively.
(Table~\ref{tab:notation} summarizes the main notations used in the paper.)

\paragraph{\bf Doubling Dimension.}
The algorithms presented in this paper are designed for general metric
spaces, and their performance is analyzed in terms of the
dimensionality of the dataset $P$, as captured by the
well-established notion of doubling dimension defined as
follows \cite{Heinonen01}.  For any $p \in P$ and $r > 0$, let the \emph{ball of radius
  $r$ centered at $p$} be the set of points of $P$ at distance at most
$r$ from $p$. The \emph{doubling dimension} of $P$ is the smallest
value $D$ such that for every $p \in P$ and $r >0$, the ball of radius
$r$ centered at $p$ is contained in the union of at most $2^D$ balls
of radius $r/2$, centered at suitable points of $P$. The doubling
dimension can be regarded as a generalization of the Euclidean
dimensionality to general spaces. In fact, it is easy to see that any
$P \subset \mathbb{R}^{\rm dim}$ under Euclidean distance has doubling
dimension $\BO{\mbox{\rm dim}}$.  

\paragraph{\bf Model of Computation.}
We present and analyze our algorithms using the \emph{MapReduce} model
of computation \cite{DeanG08,PietracaprinaPRSU12}, which is
one of the reference models for the distributed processing
of large datasets, and has been effectively used for clustering
problems (e.g., see
\cite{SreedharKR17,CeccarelloPP19,BakhthemmatI20}). A MapReduce
algorithm specifies a sequence of \emph{rounds}, where in each round,
a multiset $X$ of key-value pairs is first transformed into a new
multiset $X'$ of pairs by applying a given \emph{map function} in
parallel to each individual pair, and then into a final multiset $Y$
of pairs by applying a given \emph{reduce function} (referred to
as \emph{reducer}) in parallel to each subset of pairs of $X'$ having
the same key.  When the algorithm is executed on a distributed
platform, the applications of the map and reduce functions in each
round are (automatically) assigned to the available processors so to
maximize parallelism.  The data, maintained into a distributed storage
system, are brought to the processors' main memories in chunks, when
needed by the map and reduce functions.  Key performance indicators
are the number of rounds and the maximum local memory required by
individual executions of the map and reduce functions. Efficient
algorithms typically target few (possibly, constant) rounds and
substantially sublinear local memory.  We remark that our algorithms
can be straighforwardly rephrased for other distributed models, such
as the popular \emph{Massively Parallel Computation} (MPC) model
\cite{BeameKS13}.

\section{MapReduce algorithm for k-means with $z$ outliers} \label{sec:algorithm}
In this section, we present a MapReduce algorithm for k-means with $z$
outliers running in $\BO{1}$ rounds with sublinear local memory.  As typical of many efficient algorithms for
clustering and related problems, our algorithm uses the following
coreset-based approach. First, a suitably small weighted coreset $T$ is
extracted from the input $P$, such that each point $p \in P$ has a
``close'' proxy $\pi(p) \in T$, and the weight $w_q$ of each $q \in
T$ is the number of points of $P$ for which $q$ is proxy. Then, the
final solution is obtained by running on $T$ the best (possibly slow)
sequential approximation algorithm for weighted k-means with $z$
outliers. Essential to the success of this strategy is that $T$ can be
computed efficiently in a distributed fashion, its size is much smaller than
$|P|$, and it represents $P$ well, in the sense that: (i) the cost of any
solution with respect to $P$ can be approximated well in $T$; and (ii)
$T$ contains a good solution to $P$.

In Subsection~\ref{subsec-coreset} we describe a coreset
construction, building upon the one presented in
\cite{Har-PeledM04,MazzettoPP19} for the case without outliers, but
with crucial modifications and a new analysis needed to handle the more general cost
function, and to allow the use of bicriteria approximation algorithms
on the coreset. In Subsection~\ref{subsec-complete} we present and
analyze the final algorithm, while in Subsection~\ref{subsec-improved}
we outline how a refined coreset construction can yield  substantially lower local 
memory requirements.

\subsection{Flexible coreset construction} \label{subsec-coreset}
We first formally define two properties that capture the quality of
the coreset computed by our algorithm. Let $T$ be a subset of
$P$ weighted according to a proxy function $\pi : \; P \rightarrow T$,
where the weight of each $q \in T$ is $w_q = |\{p \in P : \; \pi(p)=q\}|$. 
\begin{definition} \label{def-approx}
For $\gamma \in (0,1)$, $(T,\mathbf{w})$ is a
\emph{$\gamma$-approximate coreset for $P$ with respect to $k$ and
  $z$} if for every $S,Z \subset P$, with $|S| \leq k$ and $|Z| \leq z$, 
we have:
\[
|\cost(P \backslash Z,S)-\cost(T,\hat{\mathbf{w}},S)| \leq \gamma \cdot \cost(P \backslash Z,S),
\]
where $\hat{\mathbf{w}}$ is such that for each $q \in T$, $\hat{w}_q = w_q-|\{p \in Z : \; \pi(p)=q\}|$.
\end{definition}

\begin{definition} \label{def-centroid}
For $\gamma \in (0,1)$, $(T,\mathbf{w})$ is a \emph{$\gamma$-centroid set for $P$
with respect to $k$ and $z$} if there exists a set $X\subseteq T$ of at most $k$ points 
such that 
\[
\cost(P\backslash\out_z(P,X),X) \leq (1+\gamma)\cdot \OPT_{k,z}(P).
\]
\end{definition}
In other words, a $\gamma$-approximate coreset can
faithfully estimate (within relative error $\gamma$) the cost of
\emph{any} solution with respect to the entire input dataset $P$,
while a $\gamma$-centroid set is guaranteed to contain \emph{one} good
solution for $P$. The following technical lemma states a sufficient condition
for a weighted set to be an approximate coreset.
\begin{lemma} \label{lem:approx}
Let $(T,\mathbf{w})$ be such that 
$\sum_{p \in P} d(p,\pi(p))^2 \leq \delta \cdot \OPT_{k,z} (P)$.
Then, $(T,\mathbf{w})$ is a $\gamma$-approximate coreset for $P$
with respect to $k$ and $z$, with $\gamma=\delta+2\sqrt{\delta}$.
\end{lemma}
\begin{proof}
Consider two arbitrary subsets $S,Z \subset P$ with $|S|=k$ and $|Z|=z$, and 
let $\hat{\mathbf{w}}$ be obtained from $\mathbf{w}$ by subtracting the contributions of the elements in $Z$ from the weights of their proxies. We have:

\begin{eqnarray*}
\lefteqn{
\left|
\cost(P \backslash Z,S)-\cost(T,\hat{\mathbf{w}},S) 
\right|  
 = |\sum_{p \in P \backslash Z} d(p,S)^2-\sum_{q \in T}\hat{w}_q d(q,S)^2 | }\\
& = & 
|\sum_{p \in P \backslash Z} d(p,S)^2-\sum_{p \in P \backslash Z} d(\pi(p),S)^2|  \\
& \leq & 
\sum_{p \in P \backslash Z}
\left| d(p,S)^2-d(\pi(p),S)^2 \right|  \\
& = & 
\sum_{p \in P \backslash Z}
(d(p,S)+d(\pi(p),S))|d(p,S)-d(\pi(p),S)|  \\
& \leq & 
\sum_{p \in P \backslash Z}
(d(p,\pi(p))+2d(p,S))d(p,\pi(p)) \\
& & (\mbox{since, by Fact~\ref{fact:triangle}}, -d(p,\pi(p) \leq d(p,S)-d(\pi(p),S) \leq d(p,\pi(p)) \\
& = & 
\sum_{p \in P \backslash Z} d(p,\pi(p))^2
+2\sum_{p \in P \backslash Z} d(p,S) \cdot d(p,\pi(p)).
\end{eqnarray*}
By the hypothesis, we have that 
$\sum_{p \in P} d(p,\pi(p))^2 \leq \delta \cdot \OPT_{k,z}(P)$,
and since $\OPT_{k,z}(P) \leq \cost(P \backslash Z,S)$,
the first sum is upper bounded by $\delta \cdot \cost(P \backslash Z,S)$.
Let us now concentrate on the second summation.
As observed in the proof of Fact~\ref{fact:triangle},
for any $a,b,c >0$, we have that $2ab \leq ca^2+(1/c)b^2$. 
Therefore,
\begin{eqnarray*}
2\sum_{p \in P \backslash Z} d(p,S) \cdot d(p,\pi(p)) 
& \leq & 
\sqrt{\delta} \sum_{p \in P \backslash Z} d(p,S)^2
+ \left({1/\sqrt{\delta}}\right) \sum_{p \in P \backslash Z} d(p,\pi(p))^2 \\
& \leq & 
\sqrt{\delta} \sum_{p \in P \backslash Z} d(p,S)^2
+ \sqrt{\delta} \cdot \OPT_{k,z}(P) \\
& \leq & 
2 \sqrt{\delta} \cdot \cost(P \backslash Z,S). 
\end{eqnarray*}
The lemma follows since $\gamma=\delta+2\sqrt{\delta}$.
\end{proof}
\sloppy
The first ingredient of our coreset construction is a primitive, called
\alg{CoverWithBalls}, which, given any set $X \subset P$, a precision
parameter $\delta$, and a distance threshold $R$, builds a weighted set $Y
\subset P$ whose size is not much larger than $X$, such that for each
$p \in P$, $d(p,Y) \leq \delta \max\{R, d(q,X)\}$.  Specifically, the
primitive identifies, for each $p \in P$, a \emph{proxy} $\pi(p) \in Y$ such
that $d(p,\pi(p)) \leq \delta \max\{R, d(p,X)\}$.  For every $q \in Y$, the returned weight $w_q$ is set 
equal to the the number of points
of $P$ for which $q$ is proxy. Primitive \alg{CoverWithBalls} has been
originally introduced in \cite{MazzettoPP19} and is based on a simple
greedy procedure. For completeness, we report the pseudocode below, as
Algorithm~\ref{alg:cover}.
\begin{algorithm}
    $Y \leftarrow \emptyset$\;
    \While{$P \neq \emptyset$}{
        $q \longleftarrow $ arbitrarily selected point in $P$\;
        $Y \longleftarrow Y \cup \{ q \}; w_q \longleftarrow 1$\;
        \ForEach{$p \in P$}{
            \If{ $d(p,q) \leq \delta \max \{R,d(p,X) \} $}{
                remove $p$ from $P$\;
                $w_q \longleftarrow w_q+1$; \{{\tt implicitly, $q$ becomes the proxy $\pi(p)$ of $p$}\}
            }         
        }
    }
    \Return $(Y,\mathbf{w})$
    \caption{\alg{CoverWithBalls}$(P,X,\delta,R)$} \label{alg:cover}
\end{algorithm}
We wish to remark that the proxy function $\pi$ is not explicitly represented and is reflected only
in the vector  $\mathbf{w}$. 
In our coreset construction, \alg{CoverWithBalls}
will be invoked multiple times to compute coresets of increasingly higher
quality. Observe that the output $(Y,\mathbf{w})$ of
\alg{CoverWithBalls}$(P,X,\delta,R)$ is implicitly associated with
a map $\pi \; : \; P \rightarrow Y$ such that:
\begin{itemize}
\item
for every $q \in Y$, $w_q = |\{p \in P : \; \pi(p)=q\}|;$
\item
for every $p \in P\backslash Y$, $d(p,\pi(p)) \leq \delta \max \{R,d(p,X) \}$.
\end{itemize}

The second ingredient of our distributed coreset construction is some
sequential algorithm, referred to as \alg{SeqkMeans} in the following,
which, given in input a dataset $Q$ and an integer $k$, computes a
$\beta$-approximate solution to the standard k-means problem \emph{without
outliers} with respect to $Q$ and $k$.

We are ready to present a 2-round MapReduce algorithm, dubbed
\alg{MRcoreset}, that, on input a dataset $P$, the values $k$ and $z$,
and a precision parameter $\gamma$, combines the two ingredients
presented above to produce a weighted coreset which is both an
$O(\gamma)$-approximate coreset and an $O(\gamma)$-centroid set with
respect to $k$ and $z$.  The computation performed by
\alg{MRcoreset}$(P,k,z,\gamma)$ in each round is described below.

\paragraph{\bf First Round} The dataset $P$ is evenly partitioned into $L$ equally sized subsets, $P_1, P_2, \ldots, P_L$, through a suitable map function. Then, in parallel, the following steps are performed by a distinct
reducer on each $P_i$, with $1\leq i\leq L$:
\begin{enumerate}
\item
\alg{SeqkMeans} is invoked with input $(P_i,k')$,
where $k'$ is a suitable
function of $k$ and $z$ that will be fixed later in the analysis,
returning a solution $S_i \subset P_i$.
\item 
Let $R_i = \sqrt{\cost(P_i, S_i)/|P_i|}$. The primitive
\alg{CoverWithBalls}$(P_i,S_i,\gamma/\sqrt{2\beta},R_i)$ is invoked, returning a weighted set of points $(C_i,\mathbf{w}^{C_i})$.
\end{enumerate}

\paragraph{\bf Second Round} 
The same partition of $P$ into $P_1, P_2, \ldots, P_L$ is used.  A
suitable map function is applied so that each reducer receives a
distinct $P_i$ and the triplets ($|P_j|$, $R_j$, $C_j$) for
all $1 \leq j \leq L$ from Round~1 (the weights $\mathbf{w}^{C_j}$ are
ignored).  Then, for $1 \leq i \leq L$, in parallel, the reducer in
charge of $P_i$ sets $R=\sqrt{\sum_{j=1}^{L} |P_j|\cdot R_j^2/|P|}$, $C =
\cup_{j=1}^{L} C_j$, and invokes \alg{CoverWithBalls}$(P_i, C, \gamma/\sqrt{2\beta},
R)$.  The invocation returns the weighted set
$(T_i,\mathbf{w}^{T_i})$. \\[0.2cm]
The final coreset returned by the
algorithm is $(T,\mathbf{w}^{T})$, where $T=\cup_{i=1}^L T_i$ and
$\mathbf{w}^{T}$ is the weight function such that $\mathbf{w}^{T_i}$
is the projection of $\mathbf{w}^{T}$ on $P_i$, for $1 \leq i \leq L$. 

We now analyze the main properties of the weighted
coreset returned by \alg{MRcoreset}, which will be exploited
in the next subsection to derive the performance-accuracy tradeoffs
featured by our distributed solution to k-means with $z$ outliers. 
Recall that we assumed that \alg{SeqkMeans} is instantiated with an
approximation algorithm that, 
when invoked on input $(P_i,k')$, returns a set
$S_i \subset P_i$ of $k'$ centers such that
$\cost(P_i,S_i) \leq \beta \cdot \OPT_{k'}(P_i)$, for some $\beta \geq 1$. 
Let $D$ denote the doubling dimension of $P$.
The following lemma is a consequence of the analysis in 
\cite{MazzettoPP19} for the case without outliers, and its proof (omitted) 
is a simple composition of the proofs of Lemmas 3.6, 3.11, and 3.12
in that paper. 
\begin{lemma} \label{lem:mazzetto}
Let $(C,\mathbf{w}^C)$ and $(T,\mathbf{w}^T)$ be the weighted coresets
computed by {\rm \alg{MRcoreset}}$(P,k,z,\gamma)$, and let
$\pi^C, \pi^T$ be the corresponding proxy functions. We have:
\[
\sum_{p \in P} d(p,\pi^X(p))^2 \leq 4 \gamma^2 \cdot \OPT_{k'}(P),
\;\;\; (\mbox{with $X = C,T$})
\]
and
\begin{eqnarray*}
|C| & = & \BO{|L|\cdot  k'\cdot (8 \sqrt{2\beta}/\gamma)^{D}\cdot \log |P|}, \\
|T| & = & \BO{|L|^2\cdot k'\cdot (8 \sqrt{2\beta}/\gamma)^{2D}\cdot \log^2 |P|}. 
\end{eqnarray*}
\end{lemma}
As noted in the introduction, while the doubling dimension $D$ appears
in the above bounds, the algorithm does not require the knowledge of
this value, which would be hard to compute.  The next theorem
establishes the main result of this section regarding the quality of
the coreset $(T,\mathbf{w}^T)$ with respect to the k-means problem
with $z$ outliers.
\begin{theorem} \label{thm:quality}
Let $\gamma$ be such that $0< \gamma \leq \sqrt{3/8}-1/2$.   By setting
$k'=k+z$ in the first round, {\rm \alg{MRcoreset}}$(P,k,z,\gamma)$
returns a weighted coreset $(T,\mathbf{w}^{T})$ which is a
$(4\gamma+4\gamma^2)$-approximate coreset and a
$27\gamma$-centroid set for $P$ with respect to $k$ and $z$.
\end{theorem}
\begin{proof}
Define $\sigma=4\gamma+4\gamma^2$ and, by the hypothesis
on $\gamma$, note that $\sigma \leq 1/2$.
The fact that $(T,\mathbf{w}^{T})$ is
a $\sigma$-approximate coreset 
for $P$ with respect to $k$ and $z$,
follows directly from Fact~\ref{fact:opt},
Lemma~\ref{lem:approx} (setting $\delta = 4\gamma^2$),
and Lemma~\ref{lem:mazzetto}.
We are left to show that 
$(T,\mathbf{w}^{T})$ is
a $27\gamma$-centroid set
for $P$ with respect to $k$ and $z$. 
Let $S^* \subset P$ be the optimal set of $k$ centers and let
$Z^* = \out_z(P,S^*)$. Hence, $\cost(P \backslash Z^*, S^*) = \OPT_{k,z}(P)$.
Define $X = \{p^T : \; p \in S^*\} \subset T$.
We show that $X$ is 
a good solution for the k-means problem with $z$ outliers for $P$. 
Clearly, $\cost(P \backslash \out_z(P,X),X) \leq \cost(P \backslash Z^*,X)$,
hence it is sufficient to upper bound the latter term.  
To
this purpose, consider the weighted set $(C,\mathbf{w}^C)$
computed at the end of Round 1, and let $\pi^C$ be the proxy function
defining the weights  $\mathbf{w}^C$. Arguing as before, we can
conclude that $(C,\mathbf{w}^C)$ is also a
$\sigma$-approximate coreset 
for $P$ with respect to $k$ and $z$. Therefore,
since $\sigma \leq 1/2$, 
\[
\cost(P \backslash Z^*,X)
\leq 
{\frac{1}{1-\sigma}} \cost(C,\hat{\mathbf{w}}^C,X)
\leq 
(1+2\sigma) \cost(C,\hat{\mathbf{w}}^C,X),
\]
where $\hat{\mathbf{w}}^C$ is obtained from $\mathbf{w}^C$ by
subtracting the contributions of
the elements in $Z^*$ from the weights of their proxies. 
Then, we have:
\begin{eqnarray*}
\cost(C,\hat{\mathbf{w}}^C,X) 
& = & \sum_{q \in C} \hat{{w}}^C_q d(q,X)^2 \\
& \leq &
(1+\gamma) \sum_{q \in C} \hat{{w}}^C_q d(q,q^{S^*})^2 +
(1+(1/\gamma))  \sum_{q \in C} \hat{{w}}^C_q d(q^{S^*},X)^2 \\
&& \mbox{(by Fact~\ref{fact:triangle})} \\
& = &
(1+\gamma) \cost(C,\hat{\mathbf{w}}^C,S^*)+
(1+(1/\gamma))  \sum_{q \in C} \hat{{w}}^C_q d(q^{S^*},X)^2 \\
& \leq &
(1+\gamma)(1+\sigma) \OPT_{k,z}(P) +
(1+(1/\gamma))  \sum_{q \in C} \hat{{w}}^C_q d(q^{S^*},X)^2 \\
&& \mbox{(since $(C,\mathbf{w}^T)$ is a $\sigma$-approximate coreset).}
\end{eqnarray*}

We now concentrate on the term $\sum_{q \in C} \hat{{w}}^C_q
d(q^{S^*},X)^2$. First observe that, since $X\subset T$ contains the point in $T$ closest to
$q^{S^*}$, we have $d(q^{S^*},X) = d(q^{S^*},T)$ and
\alg{CoverWithBalls} guarantees that $d(q^{S^*},T) \leq
(\gamma/\sqrt{2\beta}) (R+d(q^{S^*},C))$, where $R$ is the parameter
used in \alg{CoverWithBalls}.  Also, for $q \in C$, $d(q^{S^*},C) \leq
d(q^{S^*},q)$. Now,
\begin{eqnarray*}
\sum_{q \in C} \hat{{w}}^C_q d(q^{S^*},X)^2 
& \leq &
(\gamma^2/(2\beta))
\sum_{q \in C} \hat{{w}}^C_q (R^2+d(q,S^*)^2) \\
& \leq &
(\gamma^2/(2\beta)) 
\left(
((|P|-z)/|P|) \sum_{i=1}^L |P_i| \cdot R_i^2
+ \sum_{q \in C} \hat{{w}}^C_q  d(q,S^*)^2
\right) \\
& \leq &
(\gamma^2/(2\beta)) 
\left(
\sum_{i=1}^L \cost(P_i,S_i)
+ \sum_{q \in C} \hat{{w}}^C_q  d(q,S^*)^2
\right) \\
& \leq &
(\gamma^2/(2\beta)) 
\left(
\beta \sum_{i=1}^L \OPT_{k+z}(P_i)
+ \cost(C,\hat{\mathbf{w}}^C,S^*) 
\right) \\
& \leq &
(\gamma^2/2) 
\left(
\sum_{i=1}^L \OPT_{k+z}(P_i)
+ \cost(C,\hat{\mathbf{w}}^C,S^*) 
\right)
\;\; (\mbox{since $\beta \geq 1$}).
\end{eqnarray*}
Using the triangle inequality and Fact~\ref{fact:opt},
it is easy to show that $\sum_{i=1}^L \OPT_{k+z}(P_i) \leq 4 \cdot \OPT_{k,z}(P)$.
Moreover, since $(C,\mathbf{w}^C)$ is a $\sigma$-approximate coreset
for $P$ with respect to $k$ and $z$,
$\cost(C,\hat{\mathbf{w}}^C,S^*)  \leq (1+\sigma) \OPT_{k,z}(P)$. Consequently,
$\sum_{q \in C} \hat{{w}}^C_q d(q^{S^*},X)^2 \leq (\gamma^2/2)(5+\sigma) \OPT_{k,z}(P)$. Putting it all together and recalling that 
$\sigma=4\gamma+4\gamma^2 \leq 1/2$, we conclude that
\begin{eqnarray*}
\lefteqn{\cost(P \backslash Z^*,X) } \\
& \leq & {(1+2\sigma)} 
\left(
(1+\gamma)(1+\sigma) + (1+1/\gamma)(\gamma^2/2)(5+\sigma)
\right) \cdot \OPT_{k,z}(P) \\
&\leq & (1+27\gamma) \OPT_{k,z}(P).
\end{eqnarray*}
\end{proof}
\subsection{Complete algorithm} \label{subsec-complete}
Let \alg{SeqWeightedkMeansOut} be a sequential algorithm for weighted 
k-means with $z$ outliers, which, given in input a weighted set
$(T,\mathbf{w}^T)$ returns a (possibly
bicriteria) solution $S$ of $\rho k$ centers such that
$\cost(T,\hat{\mathbf{w}}^T,S) \leq \alpha \cdot
\OPT_{k,z}(T,\mathbf{w})$, where $\rho \geq 1$ and
$\hat{\mathbf{w}}^T$ is obtained from $\mathbf{w}$ by scaling $\tau z$
units of weight from the points of $T$ farthest from $S$, for some $\tau
\geq 1$. For $\gamma >0$, the complete algorithm first runs the
2-round \alg{MRcoreset}$(P,\rho k, \tau z,\gamma)$ algorithm, to
extract a weighted coreset $(T,\mathbf{w}^T)$. Then, in a third
round, the coreset is gathered in a single reducer which runs
\alg{SeqWeightedkMeansOut}$(T,\mathbf{w}^T,k,z)$ to compute the final
solution $S$. We have:

\begin{theorem} \label{thm:final}
For $0 < \gamma \leq \sqrt{3/8}-1/2$, the above 3-round MapReduce
algorithm computes a solution $S$ of at most $\rho k$ centers such
that
\[
\cost(P \backslash \out_{\tau z}(P,S),S) 
\leq (\alpha+\BO{\gamma}) \cdot \OPT_{k,z}(P),
\]
and requires 
$\BO{|P|^{2/3} \cdot (\rho k + \tau z)^{1/3} \cdot 
(8 \sqrt{2 \beta}/\gamma)^{2D}
\cdot \log^2 |P|}$ local memory.
\end{theorem}
\begin{proof}
Let $T$ be the coreset computed at Round 2, and let  $\hat{Z}\subseteq P$ be such that the scaled weight function 
 $\hat{\mathbf{w}}^T$, associated to the solution $S$ computed in Round 3, can be obtained from $\mathbf{w}^T$ by subtracting the contribution of each point in $\hat{Z}$ from the weight of its proxy in $T$. Clearly, $|\hat{Z}| \leq \tau z$ and $\cost(P \backslash \out_{\tau z}(P,S),S)\leq \cost(P \backslash \hat{Z},S)$. Now, let $\sigma=4\gamma+4\gamma^2\leq 1/2$. We know from Theorem~\ref{thm:quality} that $(T,\mathbf{w}^{T})$ is a $\sigma$-approximate coreset for $P$ with respect to $\rho k$ and $\tau z$.  We have: 
\begin{eqnarray*}
\cost(P \backslash \hat{Z},S) &\leq& \frac{1}{1-\sigma}\cost(T,\hat{\mathbf{w}}^T,S)\\
&\leq& (1+2\sigma)\cost(T,\hat{\mathbf{w}}^T,S)
\leq (1+\BO{\gamma})\cdot\alpha\cdot\OPT_{k,z}(T,\mathbf{w}).
\end{eqnarray*}
Since $\OPT_{\rho k,\tau z}(P)\leq \OPT_{k,z}(P)$, 
Fact~\ref{fact:opt} and Lemma~\ref{lem:mazzetto} can be used to
prove that 
both $(C,\mathbf{w}^{C})$ (computed in Round 1) and 
$(T,\mathbf{w}^{T})$ are $\sigma$-approximate coresets for $P$  with respect to $k$ and $z$. A simple adaptation of the proof of 
Theorem~\ref{thm:quality} shows that $(T,\mathbf{w}^{T})$ is a $27\gamma$-centroid set for $P$ with respect to $k$ and $z$. Now, let $X\subseteq T$ be the set of at most $k$ points of Definition~\ref{def-centroid}, and let $\overline{\mathbf{w}}^T$ be obtained from $\mathbf{w}^T$ by subtracting the contributions of
the elements in $\out_z(P,X)$ from the weights of their proxies. By the optimality of $\OPT_{k,z}(T,\mathbf{w})$ we have that
\begin{eqnarray*}
\OPT_{k,z}(T,\mathbf{w})&\leq&\cost(T,\overline{\mathbf{w}}^T,X) \\
&\leq & (1+\sigma)\cost(P\backslash \out_z(P,X), X)\\
&\leq &(1+\sigma)(1+27\gamma)\cdot \OPT_{k,z}(P) = (1+\BO{\gamma})\cdot\OPT_{k,z}(P).
\end{eqnarray*}
Putting it all together, we conclude that 
\begin{eqnarray*}
\cost(P \backslash \out_{\tau z}(P,S),S) \leq  \cost(P \backslash \hat{Z},S
\leq (\alpha+\BO{\gamma}) \cdot \OPT_{k,z}(P).
\end{eqnarray*}
The local memory bound follows from Lemma~\ref{lem:mazzetto}, setting
$L = (|P|/(\rho k + \tau z))^{1/3}$.
\end{proof}

\subsection{Improved local memory}\label{subsec-improved}
The local memory of the algorithm presented in the previous
subsections can be substantially improved by modifying Round 2 of
\alg{MRcoreset}$(P,k,z,\gamma)$ as follows. Now, each reducer first
determines a $\beta$-approximate solution $S_C$ to weighted k-means
(without outliers) on $(C,\mathbf{w}^C)$, with $k'=k+z$ centers, and
then runs \alg{CoverWithBalls}$(C,S_C,\gamma/\sqrt{2\beta},R)$,
yielding a weighted set $C'$, whose size is a factor $|L|$ less than
the size of $C$.  Finally, the reducer runs
\alg{CoverWithBalls}$(P_i,C',\gamma/\sqrt{2\beta},R)$. A small
adaptation to \alg{CoverWithBalls} is required in this case: when
point $p \in C$ is mapped to a proxy $q \in C'$, the weight of $q$ is
increased by $w^C_p$ rather than by one. The analysis of this modified construction is given below. 
\begin{lemma} \label{lem:improvement_C'}
Let $(C',\mathbf{w}^{C'})$ be the weighted coreset
computed by {\rm \alg{CoverWithBalls}}$(C,S_C,\gamma/\sqrt{2\beta},R)$ with $0 < \gamma \leq \sqrt{3/8}-1/2$. Then,
 there exists a proxy function $\pi^{C'}:P\to C'$ such that
\[
\sum_{p \in P} d(p,\pi^{C'}(p))^2 \leq 18 \gamma^2 \cdot \OPT_{k'}(P).
\]
\end{lemma}
\begin{proof}
Let $\pi^{C}:P\to C$ be the proxy function of Lemma~\ref{lem:mazzetto}, and let $\phi^{C'}:C\to C'$ be the map induced by \alg{CoverWithBalls}$(C,S_C,\gamma/\sqrt{2\beta},R)$. Define $\pi^{C'}:P\to C'$ as $\phi^{C'}\circ\pi^{C}$. By the triangle inequality and Lemma~\ref{lem:mazzetto} we have that
\begin{eqnarray*}
\sum_{p \in P} d(p,\pi^{C'}(p))^2&\leq&2\sum_{p \in P} d(p,\pi^{C}(p))^2+2\sum_{p \in P} d(\pi^{C}(p),\pi^{C'}(p))^2\\
&\leq&8\gamma^2\cdot\OPT_{k'}(P)+2\sum_{q \in C} w_q^Cd(q,\phi^{C'}(q))^2.
\end{eqnarray*}
The latter term can be bounded by using the properties of \alg{CoverWithBalls} as follows. Let $\overline{S}^*$ be the optimal centers for $P$ with respect to $k'$. We have that
\begin{eqnarray*}
\sum_{q \in C} w^C_q d(q,\phi^{C'}(q))^2 
& \leq &
(\gamma^2/(2\beta))
\sum_{q \in C} w^C_q (R^2+d(q,S_C)^2) \\
& \leq &
(\gamma^2/(2\beta)) 
\left(
\sum_{i=1}^L |P_i| \cdot R_i^2
+ \cost(C,\mathbf{w}^C,S_C)\right) \\
& \leq &
(\gamma^2/(2\beta)) 
\left(
\sum_{i=1}^L \cost(P_i,S_i)
+ \cost(C,\mathbf{w}^C,S_C)
\right) \\
& \leq &
(\gamma^2/(2\beta)) 
\left(
4\beta \sum_{i=1}^L \cost(P_i,\overline{S}^*)
+ \beta \cdot\OPT_{k'}(C,\mathbf{w}^C)\right)  \\
& \leq &
(\gamma^2/2) 
\left(
4\cdot\OPT_{k'}(P)
+ \OPT_{k'}(C,\mathbf{w}^C)
\right).
\end{eqnarray*}
By combining the arguments of  Lemma~\ref{lem:approx}and Lemma~\ref{lem:mazzetto},  we obtain that $(C,\mathbf{w}^C)$ is a 
$\sigma$-approximate coreset for $P$ with respect to $k'$ and $z=0$, with $\sigma=4\gamma+4\gamma^2\leq 1/2$. Thus, using again the triangle inequality
\begin{eqnarray*}
\OPT_{k'}(C,\mathbf{w}^C)&\leq&4\cost(C,\mathbf{w}^C,\overline{S}^*)\\
&\leq&4(1+\sigma)\cost(P,\overline{S}^*)\\
&\leq&6\cdot\OPT_{k'}(P).
\end{eqnarray*}
Putting it all together, we conclude that 
\begin{eqnarray*}
\sum_{p \in P} d(p,\pi^{C'}(p))^2 \leq 18 \gamma^2 \cdot \OPT_{k'}(P).
\end{eqnarray*}
\end{proof}

\begin{lemma} \label{lem:improvement_T}
Let $(T,\mathbf{w}^T)$ be the weighted coreset
computed by {\rm \alg{MRcoreset}}$(P,k,z,\gamma)$, and let
$\pi^T$ be the corresponding proxy function. We have:
\[
\sum_{p \in P} d(p,\pi^T(p))^2 \leq 11 \gamma^2 \cdot \OPT_{k'}(P).
\]
\end{lemma}
\begin{proof}
By the properties of the output of \alg{CoverWithBalls} and Lemma~\ref{lem:improvement_C'}, we have that
\begin{eqnarray*}
\sum_{p \in P} d(p,\pi^T(p))^2
& \leq &
(\gamma^2/(2\beta)) 
\left(
\sum_{i=1}^L |P_i| \cdot R_i^2
+ \sum_{p \in P} d(p,\pi^{C'}(p))^2 \right) \\
& \leq &
(\gamma^2/(2\beta)) 
(
4\beta+18\gamma^2)\cdot\OPT_{k'}(P)
\\
&\leq&
11\gamma^2\cdot\OPT_{k'}(P).
\end{eqnarray*}
\end{proof}

From Fact~\ref{fact:opt} and Lemmas~\ref{lem:approx} and~\ref{lem:improvement_T}, it follows that $(T,\mathbf{w}^T)$ is an 
$(11\gamma^2+2\sqrt{11}\gamma)$-approximate coreset for $P$ with respect to $k$ and $z$. Moreover,  by a slight adaptation 
of the proof of Theorem~\ref{thm:quality} and by setting $\gamma \leq (\sqrt{3}-\sqrt{2})/6$, we have that $(T,\mathbf{w}^T)$ is also a 
$47\gamma$-centroid set for $P$ with respect to $k$ and $z$. 

The main result is stated in the following theorem.

\begin{theorem} \label{thm:final_improved}
For $0 < \gamma \leq (\sqrt{3}-\sqrt{2})/6$, the modified
3-round MapReduce algorithm computes a solution $S$ of at most $\rho
k$ centers such that
\[
\cost(P \backslash \out_{\tau z}(P,S),S) 
\leq (\alpha+\BO{\gamma}) \cdot \OPT_{k,z}(P).
\]
The algorithm requires 
$\BO{|P|^{1/2} \cdot (\rho k + \tau z)^{1/2} \cdot 
(8 \sqrt{2 \beta}/\gamma)^{2D}
\cdot \log^2 |P|}$ local memory.
\end{theorem}
\begin{proof}
The bound on the approximation factor is obtained as a straightforward adaptation of the proof of Theorem~\ref{thm:final}. For what concerns the local memory requirements, the result in \cite{MazzettoPP19} concerning the size of the output of \alg{CoverWithBalls} yields:
\begin{eqnarray*}
|C'| & = & \BO{(\rho k + \tau z)\cdot (8 \sqrt{2\beta}/\gamma)^{D}\cdot \log |P|}, \\
|T| & = & \BO{|L|\cdot (\rho k + \tau z)\cdot (8 \sqrt{2\beta}/\gamma)^{2D}\cdot \log^2 |P|}.
\end{eqnarray*}
Setting $L = (|P|/(\rho k + \tau z))^{1/2}$, we get the desired upper bound.
\end{proof}

\section{Instantiation with different sequential algorithms for weighted k-means} \label{sec-wkmeans}

We briefly outline how to adapt two state-of-the-art sequential
algorithms for k-means with $z$ outliers in general metrics, namely,
\alg{LS-Outlier} by \cite{GuptaKLMV17} and \alg{k-Means-Out} by
\cite{FriggstadKRS19}, to handle the weighted variant of the
problem. The algorithms are bicriteria, in the sense that
the approximation guarantee is obtained at the expense of
a larger number of outliers (\alg{LS-Outlier}), or a
larger number of centers (\alg{k-Means-Out}). Then, we assess the
accuracy-resource tradeoffs attained by the MapReduce algorithm of
Section~\ref{sec:algorithm}, when these algorithms are employed in its
final round.

Given a set of points $P$ and parameters $k$ and $z$,
\alg{LS-Outlier} starts with a set $C \subset P$ of $k$ arbitrary
centers and a corresponding set $Z = \out_z(P,C)$ of outliers. Then,
for a number of iterations, it refines the selection $(C,Z)$ to
improve the value $\cost(P \backslash Z,C)$ by a factor at least
$1-\epsilon/k$, for a given $\epsilon >0$, until no such improvement
is possible.  In each iteration, first a new set $C'$ is computed
through a standard local-search \cite{KanungoMNPSW04} on $P \backslash
Z$, and then a new pair $(C_{\rm new},Z_{\rm new})$ with minimal
$\cost(P \backslash Z_{\rm new},C_{\rm new})$ is identified among the
following ones: $(C',Z \cup \out_z(P \backslash Z,C')$ and
$(C'',Z \cup \out_z(P,C'')$, where $C''$ is obtained from $C'$ 
with the most profitable swap between
a point of $P$ and a point of $C'$.

It is shown in \cite{GuptaKLMV17} that \alg{LS-Outlier} returns a pair
$(C,Z)$ such that $\cost(P \backslash Z,C) \leq 274 \cdot
\OPT_{k,z}(P)$ and $|Z| = \BO{(1/\epsilon)kz \log (|P| \Delta)}$, where $\Delta$
is the ratio between the maximum and minimum pairwise distances in
$P$. \alg{LS-Outlier} can be adapted for the weighted variant of the
problem as follows.  Let $(P,\mathbf{w})$ denote the input
pointset.
In this weighted setting, the role of a set $Z$ of $m$ outliers is played by a 
weight function $\mathbf{w}^Z$ such that $0 \leq w^Z_p \leq w_p$,
for each $p \in P$, and $\sum_{p\in P} w^Z_p = m$. The union of two
sets of outliers in the original algorithm is replaced by the
pointwise sum or pointwise maximum of the corresponding weight functions,
depending on whether the two sets are disjoint (e.g., $Z$ and
$\out_z(P \backslash Z,C')$) or not (e.g., $Z$ and $\out_z(P,C'')$).
It can be proved that with this adaptation the algorithm returns a
pair $(C,\mathbf{w}^Z)$ such that $\cost(P,
\mathbf{w}-\mathbf{w}^Z,C) \leq 274 \cdot \OPT_{k,z}(P,\mathbf{w})$
and $\sum_{p\in P} w^Z_p = \BO{(1/\epsilon)kz \log (|P| \Delta)}$.

Algorithm \alg{k-Means-Out} also implements a local search.  For given
$\rho, \epsilon >0$, the algorithm starts from an initial set $C
\subset P$ of $k$ centers and performs a number of iterations, where
$C$ is refined into a new set $C'$ by swapping a subset $Q \subset C$
with a subset $U \subset P \backslash C$ (possibly of different size),
such that $|Q|,|U| \leq \rho$ and $|C'| \leq (1+\epsilon)k$, as long
as $\cost(P \backslash \out_z(P,C'),C') < (1-\epsilon/k) \cdot
\cost(P \backslash \out_z(P,C),C)$.  It is argued in
\cite{FriggstadKRS19} that for $\rho =
(D/\epsilon)^{\BT{D/\epsilon}}$, \alg{k-Means-Out} returns a set $C$
of at most $(1+\epsilon)k$ centers such that $\cost(P \backslash
\out_z(P,C),C) \leq (1+\epsilon) \cdot \OPT_{k,z}(P)$, where $D$ is
the doubling dimension of $P$. The running time is exponential in
$\rho$, so the algorithm is polynomial when $D$ is constant.

The adaptation of \alg{k-Means-Out} for the weighted variant for an
input $(P,\mathbf{w})$ is straightforward and concerns the cost
function only. It is sufficient to substitute $\cost(P \backslash
\out_z(P,C),C)$ with $\cost(P,\hat{\mathbf{w}},C)$, where
$\hat{\mathbf{w}}$ is obtained from $\mathbf{w}$ by scaling the
weights associated with the points of $P$ farthest from $C$,
progressively until exactly $z$ units of weights overall are
subtracted.  It can be proved that with this adaptation the algorithm
returns a set $C$ of at most $(1+\epsilon)k$ centers such that
$\cost(P,\hat{\mathbf{w}},C) \leq (1+\epsilon) \cdot \OPT_{k,z}(P)$.
 
By Theorems~\ref{thm:final} and~\ref{thm:final_improved}, these two sequential strategies can be
invoked in Round 3 of our MapReduce algorithm to yield bicriteria
solutions with an additive $\BO{\gamma}$ term in the approximation
guarantee, for any sufficiently small $\gamma >0$.

\section{Conclusions}\label{sec:conclusions}
We presented a flexible, coreset-based framework able to yield a
scalable, 3-round MapReduce algorithm for k-means with $z$ outliers,
featuring an approximation quality which can be made arbitrarily close
to the one of any sequential (bicriteria) approximation algorithm for
the weighted variant of the problem, and requiring local memory
substantially sublinear in the size of the input dataset, when this
dataset has bounded dimensionality.  Our approach naturally extends to
the case of k-median clustering with outliers, where the cost function
is the sum of distances rather than of squared distances.  Future
research will target the adaptation of the state-of-the-art
non-bicriteria LP-based algorithm of \cite{KrishnaswamyLS18} to the
weighted case, and the generalization of our approach to other
clustering problems.

%
%
%
\bibliographystyle{plain}

\end{document}